\documentclass[11pt,reqno]{amsart}

\usepackage{amsmath,amssymb}
\usepackage{graphicx}
\usepackage[pagebackref, colorlinks = true, linkcolor = blue, urlcolor  = blue, citecolor = red]{hyperref}
\usepackage{xcolor}

\usepackage[margin=1in]{geometry}

\usepackage{caption} 
\usepackage{subcaption} 
\usepackage{enumitem}

\newcommand{\eq}[1]{\begin{align}#1\end{align}}
\newcommand{\eqb}[1]{\begin{align}\begin{aligned}#1\end{aligned}\end{align}}

\newtheorem{theorem}{Theorem}[section]
\newtheorem*{theorem*}{Theorem}

\newtheorem{corollary}[theorem]{Corollary}
\newtheorem{lemma}[theorem]{Lemma}
\newtheorem{remark}[theorem]{Remark}

\DeclareMathOperator{\trace}{Tr} 

\DeclareMathOperator{\mob}{\textnormal{M\"ob}}

\begin{document}

\date{\today}

\title[On the minimum output entropy of random orthogonal quantum channels]{On the minimum output entropy of \\ random orthogonal quantum channels}

\author{Motohisa Fukuda}
\address{MF: Yamagata University, 1-4-12 Kojirakawa, Yamagata, 990-8560 Japan}
\email{fukuda@sci.kj.yamagata-u.ac.jp}

\author{Ion Nechita}
\address{IN: Zentrum Mathematik, M5, Technische Universit\"at M\"unchen, Boltzmannstrasse 3, 85748 Garching, Germany
and CNRS, Laboratoire de Physique Th\'{e}orique, IRSAMC, Universit\'{e} de Toulouse, UPS, F-31062 Toulouse, France}
\email{nechita@irsamc.ups-tlse.fr}

\subjclass[2000]{}

\begin{abstract}
We consider sequences of random quantum channels defined using the Stinespring formula  with Haar-distributed random orthogonal matrices. 
For any fixed sequence of input states, we study the asymptotic eigenvalue distribution of the outputs through tensor powers of random channels. 
We show that the input states achieving minimum output entropy are tensor products of maximally entangled states (Bell states) when the tensor power is even. 
This phenomenon is completely different from the one for random quantum channels constructed from Haar-distributed random unitary matrices,
which leads us to formulate some conjectures about the regularized minimum output entropy.
\end{abstract}

\maketitle

\tableofcontents

\section{Introduction}\label{seq:intro}

One of most important questions in quantum information theory is to determine the optimal rate of transmission of classical information through noisy quantum channels. Unlike its classical counterpart, no closed formula has been found yet for the classical capacity of quantum channels. 
Since the capacity is defined as the maximum rate at which classical information can be sent reliably over the channel in a way that the probability of error approaches zero as the length of codes goes infinity, naturally the capacity $C(\cdot)$ of a quantum channel $\Phi$ has an asymptotic formula \cite{holevo1998capacity,schumacher1997sending}
\eq{\label{eq:capacity}
C(\Phi) =\lim_{r \to \infty} \frac 1r \chi(\Phi^{\otimes r})
}
where $\chi(\cdot)$ is the Holevo capacity. 
Here, we assume that the errors appearing in the transmission of information are independent along the uses of the quantum channels $\Phi$, and  it is represented by the tensor power in the formula. 

For some classes of channels, such as depolarizing channels \cite{king2003capacity}, entanglement breaking channels \cite{shor2002additivity, King2003e}, Hadamard channels \cite{king2005properties}, and unital qubit channels \cite{king2002additivity}, the above formula \eqref{eq:capacity} can be simplified. 
This is a consequence of the following additivity property proved in the above cited papers: for any $r \in \mathbb N$
\eq{\label{eq:chi-additivity}
\chi(\Phi^{\otimes r}) = r \chi(\Phi).
}
Additivity for the Holevo capacity yields a closed formula (called a \emph{single-letter formula}) for the classical capacity for such channels: $C(\Phi) = \chi(\Phi)$.

However, the above simplification does not hold for all quantum channels. 
In a breakthrough paper \cite{hastings2009superadditivity}, Hastings showed violation of additivity for another quantity, the \emph{minimum output entropy}, 
which implies that \eqref{eq:chi-additivity} does not hold for some quantum channels.
These two concepts of minimum output entropy and Holevo capacity are originally different;
the former only cares about single output states, while the latter deals with ensembles of outputs (see Section \ref{sec:qit} for the exact definitions).
However, previous to Hastings' work, Shor showed \cite{shor2004equivalence} that additivity properties for those two quantities are globally equivalent to each other, allowing the translation of counter-examples from one setting to the other. 

In this paper, we focus on the minimum output entropy $S_{\min}(\Phi^{\otimes r})$,
which has close conceptual connection to $\chi (\Phi^{\otimes r})$. 
We inquire what kind of inputs states will minimize the output entropy for randomly chosen quantum channels. We explain briefly our methodology in three main points. 

First, we choose to focus on \emph{random quantum channels}. The interest in the study of random quantum channels comes mainly from the fact that, to date,  violation of additivity is proved 
only through random techniques (typically with random unitary quantum channels generated by random unitary matrices), see \cite{hastings2009superadditivity,fukuda2010comments,fukuda2010entanglement,aubrun2011hastingss,belinschi2012eigenvectors,fukuda2014revisiting,belinschi2016almost,collins2016haagerup}.
Non-random counter-examples have been obtained only for $p$-R\'enyi minimum output entropies, see \cite{werner2002counterexample,grudka2010constructive}.

Second,  our main results concern \emph{random orthogonal quantum channels}. As is explained in Section \ref{sec:qit}, any quantum channel can be dilated to a unitary closed evolution on a larger space. In this work, we only consider the case where closed dynamics comes from an orthogonal rotation. The reason for this choice is that it allows us to consider identical copies of a random quantum channel, whereas if one uses the more general unitary evolutions, then one needs to take pairs of a channel and its complex conjugate to witness additivity violations:
\eq{
S_{\min}(\Phi \otimes \bar \Phi) < S_{\min}(\Phi)+S_{\min}( \bar \Phi)
}
where the complex conjugation are applied to the unitary matrix which defines the channel $\Phi$.
To translate this result into a violation inequality for two copies of the same channel
\eq{
 S_{\min}(\Phi^{\otimes 2}) < 2S_{\min}(\Phi) 
 }
one needs to restrict themselves to the real case, where the complex conjugate does not make any difference 
(unless one employs a particular symmetrization operation, see \cite{fukuda2007simplifying}). 

Third, we shall {fix a sequence of input states}, and study the asymptotic behavior of the output states. 
In order to obtain the exact value of the minimum output entropy, one has to optimize over all input states for a fixed realization of the random quantum channel, 
but our current techniques do not allow this setting. This is indeed a drawback of our method, but in this setting we can obtain quite precise results on the possible outputs in the asymptotic limit. The current setting, where a \emph{universal, channel-independent} encoding is considered, is related to the coding theory for \emph{compound quantum channels}, see e.g.~\cite{datta2007coding,bjelakovic2009classical,mosonyi2015coding}.

Our main results (Theorem \ref{thm:optimal-input-seq} and Corollary \ref{cor:min-entropy}) can be informally stated as follows. 

\begin{theorem*}
Consider random quantum channels $\Phi_n$ obtained by partial-tracing the action of Haar-distributed random orthogonal matrices,  where $n$ is the system dimension.
Then, among fixed sequences of input states, the ones achieving minimum output entropy (asymptotically, as $n \to \infty$) for the channels $\Phi_n^{\otimes 2r}$ 
are tensor products of $r$ maximally entangled states (Bell states).
\end{theorem*}

The paper is organized as follows. In Sections \ref{sec:qit} and \ref{sec:combinatorics} we recall, respectively, some basics notions and facts from quantum information theory and from the combinatorial theory of permutations and pairings. In Section \ref{sec:Wg-O} we present the theory of invariant integration over the orthogonal group, using the graphical tensor notation. We discuss then in Section \ref{sec:outputs} the model of random quantum channels we are studying. Sections \ref{sec:outputs} and \ref{sec:optimal-inputs} are the technical core of the paper, in which we characterize the asymptotical output states for an arbitrary fixed sequence of inputs, and then we optimize over input sequences. Finally, we discuss our results and a few conjectures in the closing Section \ref{sec:discussion}.

\medskip

\noindent {\it Acknowledgement.} We would like to thank the referees for their very helpful comments which helped improve the quality of the presentation. I.N.'s research has been supported by a von Humboldt fellowship, the ANR project {StoQ} {ANR-14-CE25-0003-01}.
M.F. was financially supported by JSPS KAKENHI Grant Number JP16K00005.
I.N and M.F. are both supported by the PHC Sakura program (project number: 38615VA), implemented by the French Ministry of Foreign Affairs, the French Ministry of Higher Education and Research and the Japan Society for Promotion of Science. Both authors acknowledge the hospitality of the TU M\"unchen, where this research was conducted.

\section{Basics from quantum information theory}\label{sec:qit}

We review in this section some basic definitions and facts from quantum information theory. Some excellent references on the subject are \cite{nielsen2010quantum} and \cite{wilde2017quantum}. 

A quantum state is a positive semidefinite matrix with unit trace; we denote the set of quantum states by
$$\mathcal M_d^{1,+}(\mathbb C) := \{ \rho \in \mathcal M_d(\mathbb C)\, : \, \rho \geq 0 \text{ and } \operatorname{Tr} \rho = 1\}.$$
Rank one projections $\rho = xx^*$ (here, $x\in \mathbb C^d$, $\|x\|=1$) are the extremal points of the convex body of quantum states. In the case of bipartite composite systems, the state space is the tensor product $[\mathcal M_{d_1}(\mathbb C) \otimes \mathcal M_{d_2}(\mathbb C)]^{1,+}$. Of particular importance is the \emph{maximally entangled state} $\hat\omega = d^{-1} \Omega\Omega^* \in \mathcal M_{d^2}^{1,+}(\mathbb C)$, 
which is also called \emph{Bell state}. Here, 
$$\mathbb C^d \otimes \mathbb C^d \ni \Omega :=  \sum_{i=1}^d e_i \otimes e_i$$
is a vector of norm $\sqrt d$ (hence the normalization factor $d^{-1}$  in the formula for $\hat\omega$). We denote by $\omega = \Omega \Omega^*$ the un-normalized version of $\hat \omega$. One can extend, using functional calculus, the notion of (Shannon) entropy to quantum states:
$$S(\rho)  = - \operatorname{Tr} \rho \log \rho,$$
a quantity which is called the \emph{von Neumann entropy} of the quantum state $\rho$. 

\emph{Quantum channels} are the most general transformations of quantum states allowed by the laws of quantum mechanics. Mathematically, quantum channels are completely positive, trace preserving maps between two matrix algebras (remember that we are concerned here only with finite-dimensional quantum systems). By the celebrated Stinespring dilation theorem \cite{stinespring1955positive}, all quantum channels $\Phi:\mathcal M_d(\mathbb C) \to \mathcal M_k(\mathbb C)$ can be obtained as 
$$\Phi(X) = [\operatorname{id} \otimes \operatorname{Tr}](VXV^*),$$
where $V :\mathbb C^d \to \mathbb C^k \otimes \mathbb C^n$ is an isometry, and $n$ is a parameter (called the \emph{ancilla dimension}) which can be taken to be $n = dk$. 

As explained in the introduction, quantum Shannon theory is concerned with information transmission tasks in the quantum world. One of the fundamental information processing protocols is the transmission of classical information through a noisy quantum channel. The \emph{classical capacity} of a quantum channel $\Phi$ is defined as the optimal rate (\# bits transmitted) / (\# uses of channel), assuming that the probability of successfully decoding the transmitted information approaches one. The mathematical theory was developed in \cite{holevo1998capacity} and \cite{schumacher1997sending}, see also \cite[Section 20]{wilde2017quantum} for a textbook presentation. The definition of the classical capacity of a given quantum channel $\Phi$ is
$$C(\Phi) = \lim_{r \to \infty} \frac 1 r \chi (\Phi^{\otimes r}),$$ 
where $\chi$ is the Holevo capacity of $\Phi$ given by
$$\chi(\Phi) = \max_{\{p_i, \rho_i\}} S(\Phi(\sum_i p_i \rho_i)) - \sum_i p_i S(\Phi(\rho_i)),$$
where the maximum is taken over all ensembles of probability weights $p_i$ and input quantum states $\rho_i$ (actually, ensembles of size $d^2$, where $d$ is the dimension of the input space of $\Phi$ are enough). 

The question whether the quantity $\chi$ is additive, i.e.
$$\forall \Phi, \Psi, \qquad \chi(\Phi \otimes \Psi) = \chi(\Phi) + \chi(\Psi)$$
is known as the \emph{additivity problem} \cite{king2001minimal}. Shor has shown in \cite{shor2004equivalence} that the additivity of $\chi$ is equivalent to the additivity of a much simpler quantity, the \emph{minimum output entropy}
$$S_{\min}(\Phi) = \min_{\rho \in \mathcal M_d^{1,+}(\mathbb C)} S(\Phi(\rho)).$$
Much of the work on the additivity problem was about the quantity $S_{\min}$, proving either that additivity holds for particular classes of channels, 
or providing counter-examples (see discussion and references in Section \ref{seq:intro}). 
The focus of the current paper is to understand, for a random orthogonal quantum channel $\Phi$, 
how additivity $S_{\min}(\Phi^{\otimes r}) = r S_{\min}(\Phi)$ is violated and to find input states achieving $S_{\min}(\Phi^{\otimes r})$.

\section{Combinatorial aspects of permutations and pairings}\label{sec:combinatorics}

As the reader shall see in the next section, the theory of invariant integration over the orthogonal group $\mathcal O(d)$ is intimately connected to the combinatorial theory of pairings and permutations. We gather in the current section the necessary definitions and basic facts from combinatorics, as well as some useful lemmas. 

We denote by $\mathcal S_r$ the symmetric group on $r$ elements. 
For a permutation $\alpha \in \mathcal S_r$, we denote by $\#\alpha$ the number of its cycles (including fixed points). The quantity $|\alpha| = r - \#\alpha$ is called the \emph{length} of $\alpha$, and it can be shown to be equal to the minimal number of transpositions that multiply to $\alpha$. Also, $|\alpha|$ is the distance between $\alpha$ and the identity permutation $\mathrm{id} \in \mathcal S_r$ inside the Cayley graph of $\mathcal S_r$ generated by all transpositions. 
Permutations $\alpha,\beta,\gamma \in S_r$ satisfy triangle inequality: $|\alpha  \beta^{-1}| \leq |\alpha  \gamma^{-1}| + |\gamma  \beta^{-1}|$,
and when the equality holds, we say that $\gamma$ is on a \emph{geodesic} connecting $\alpha$ and $\beta$, and express it as 
\eq{\label{eq:geodesic-notation}
\alpha - \gamma - \beta
} 
We write $\tilde{\mathcal S}_{2r}$ for the set of products of $r$ disjoint transpositions. The set $\tilde{\mathcal S}_{2r}$ is in bijection with the set of pairings of $[2r]:=\{1, 2, \ldots, 2r\}$. To any permutation $\alpha \in \mathcal S_r$, we associate an unoriented graph $G_\alpha$, which has vertex set $V = [r]$ and edge set $E = \{ \{i, \alpha(i)\} \, : \, i \in [r] \}$. It is obvious that each vertex has degree 2 (a loop at a vertex contributes degree 2 to that vertex) and that the cycles of $\alpha$ are in bijection with the connected components of $G_\alpha$. In particular, it holds that $G_\alpha$ has $\#\alpha$ connected components. We investigate next a similar setting, where the permutation is replaced by a set of pairings. 

To a pair $(\alpha,\beta)$ of pairings of the set $[2r]$, encoded by permutations $\alpha,\beta \in \tilde{\mathcal S}_{2r}$, we associate an unoriented graph $G_{\alpha,\beta}$ having vertex set $V=[2r]$, and edge set given by 
$$E = \{ \{i, \alpha(i)\} \, : \, i \in [2r] \} \cup \{ \{i, \beta(i)\} \, : \, i \in [2r] \},$$
with the convention that we allow multiple (in our case, at most 2) edges between two vertices. The following lemma is implicit in \cite[Lemma 3.5]{collins2006integration}
\begin{lemma}\label{lem:connected-components-pairings}
The number of connected components of the graph $G_{\alpha,\beta}$ is $\#(\alpha \beta)/2 = r - |\alpha \beta|/2$.
\end{lemma}
\begin{proof}
First, note that $\#\theta  =2 r - |\theta|$ for $\theta \in \mathcal S_{2r}$.
Indeed, choose $\gamma \in \mathcal S_{2r}$ so that $\gamma\tau\gamma^{-1}$ is non-crossing, but it implies that
\begin{align}
\#\theta= \#(\gamma \theta \gamma^{-1})  = 2r - |\gamma \theta\gamma^{-1}| = 2r-  |\theta|
\end{align}
based on the well-known fact on non-crossing permutations \cite{NicaSpeicher}. 

Next, we count the number of connected components of $G_{\alpha,\beta}$ for $\alpha,\beta \in \tilde S_{2r}$.
To do so, we analyze the connected component which includes $1$. 
Suppose a number, say, $m$ is connected to $1$ in the graph $G_{\alpha,\beta}$. 
Then, we have the following two exclusive cases.
\eqb{
1 \mapsto \beta(1) \mapsto \alpha\beta(1) \mapsto \beta\alpha\beta(1) \mapsto &\ldots \mapsto  m \\
1 \mapsto \alpha(1) \mapsto \beta\alpha(1) \mapsto \alpha\beta\alpha(1) \mapsto& \ldots\mapsto m,
}
i.e. we can reach $m$ by applying $\alpha$ and $\beta$ in turn because of the idempotent property: 
$\alpha^2  = \mathrm{id}= \beta^2$. 
Hence, we now have identified the connected component which includes $1$ as a disjoint union of two sets of vertices:
\eq{\label{eq:two-part}
 \{ (\alpha\beta)^l(1): l \in \mathbb Z\} \sqcup \{(\alpha\beta)^l \alpha (1): l \in \mathbb Z\}
}
Indeed, we have
\eqb{
\beta\alpha &= \beta^{-1} \alpha^{-1}= (\alpha\beta)^{-1} \\
\beta &= \beta \alpha \alpha = (\alpha\beta)^{-1} \alpha
}
Hence, a connected component in the graph $G_{\alpha,\beta}$ always consists of two loops 
generated by $\alpha\beta(i)$ and $\beta\alpha(i)$ for some $i \in [2r]$,
so that the number of connected components is $\frac{\#(\alpha\beta)}{2}$.
In fact,
\eq{
\alpha(1) = (\alpha\beta)^l \alpha(1)  = \alpha (\alpha \beta)^{-l}(1)  \qquad \Leftrightarrow \qquad (\alpha \beta)^l (1) =1
}
This completes the proof.
\end{proof}
To understand the proof more intuitively see Figure \ref{fig:connected-component}.
All numbers connected to $1$ are represented by black and white dots, where from left to right
$1\mapsto \beta(1) \mapsto \alpha\beta(1) \mapsto \ldots \mapsto (\alpha\beta)^l(1) =1$ for some $l$.
The left part of \eqref{eq:two-part} corresponds to the black dots and the right the white dots. 
Note that $\alpha(1) = \beta (\alpha\beta)^{l-1}(1) = (\beta \alpha)^{l-1} \beta(1)$
and those arrows represent applications of $\alpha\beta$.

\begin{figure}
\includegraphics{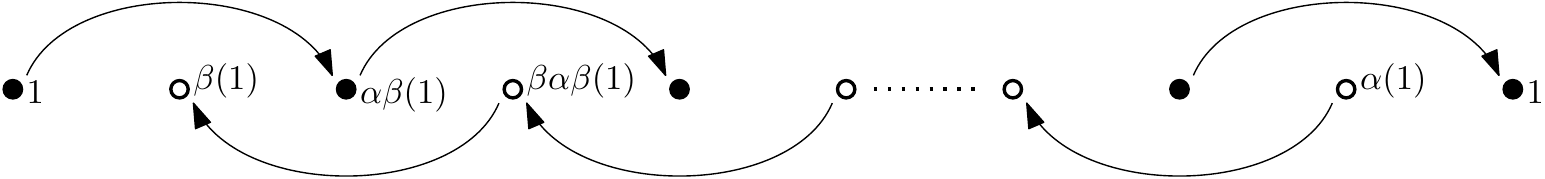}
\caption{The connected component with  $1$.
The black dots and their associated arrows show how $\alpha\beta$ forms a loop starting with $1$,
and the white ones with $\alpha(1)$.} 
\label{fig:connected-component}
\end{figure}
\section{Invariant integration over the orthogonal group}\label{sec:Wg-O}

Since the technical core of the paper consists of moment computation for random, Haar distributed orthogonal matrices, we review in this section the Weingarten formula for averaging over the orthogonal group. 

Following the work of Weingarten \cite{weingarten1978asymptotic}, the modern mathematical formulation was developed by Collins and \'Sniady in \cite{collins2006integration}; some further elements can be found in \cite{collins2009some,banica2010orthogonal}. The orthogonal Weingarten formula provides a combinatorial expression for the average of a monomial in the entries of a Haar orthogonal matrix. 

\begin{theorem}\cite[Corollary 3.4]{collins2006integration}
For every choice of indices $i_1, \ldots, i_{2r}$ and $j_1, \ldots, j_{2r}$, we have
\begin{equation}\label{eq:Wg}
\int_{\mathcal O(n)} U_{i_1j_1}\cdots U_{i_{2r}j_{2r}} dU = \sum_{\alpha, \beta \in \tilde{\mathcal S}_{2r}} \prod_{s=1}^{2r}\delta_{i_s,i_{\alpha(s)}}\delta_{j_s,j_{\beta(s)}} \operatorname{Wg}_n(\alpha,\beta).
\end{equation}
The odd moments vanish:
$$\int_{\mathcal O(n)} U_{i_1j_1}\cdots U_{i_{2r+1}j_{2r+1}} dU = 0.$$
\end{theorem}

The Weingarten function $\operatorname{Wg}$ is a combinatorial function, which can either be seen as the matrix inverse of the loop counting matrix in the Brauer algebra or as a sum over Young diagrams, see \cite{collins2006integration}. The values of this function for $r \leq 4$ can be found in \cite[Section 6]{collins2006integration}. In \cite[Theorem 3.13]{collins2006integration}, the authors also compute the leading order in the large $n$ asymptotic expansion of the orthogonal Weingarten function:
\begin{equation}\label{eq:Wg-asympt}
\operatorname{Wg}_n(\alpha,\beta) = (1+o(1))n^{-r-|\alpha\beta|/2}\operatorname{\textnormal{M\"ob}}(\alpha,\beta),
\end{equation}
where $\operatorname{\textnormal{M\"ob}}$ is the M\"obius function that we define next (see \cite[Section 3.3]{collins2006integration}). Let $2p_i$ be the number of cycles of the permutation $\alpha\beta$ having length $i$ (this number is indeed even, see Lemma \ref{lem:connected-components-pairings}). Then, define
\begin{equation}\label{eq:Mob}
\operatorname{\textnormal{M\"ob}}(\alpha,\beta):=\prod_i (-1)^{p_i-1}\operatorname{Cat}_{p_i-1},
\end{equation}
where $\operatorname{Cat}_p$ is the $p$-th Catalan number 
$$\operatorname{Cat}_p = \frac{1}{p+1} \binom{2p}{p}.$$

In \cite{collins2010random} and \cite{collins2011gaussianization}, the authors introduced a \emph{graphical calculus} for computing expectation values of expressions involving random unitary matrices and, respectively, random Gaussian matrices. We present next an natural extension of these ideas to integrals over the orthogonal group with respect to the Haar measure. We shall be brief in our exposition, since the procedure is very similar to the one in \cite{collins2010random}, also described at length in \cite[Section III.C]{collins2016random}. We shall encode tensors (i.e.~vectors, linear forms, matrices, bipartite matrices, etc.) by boxes having labels attached to them corresponding to the respective vector spaces. Empty labels are associated to duals of vector spaces (linear forms, or ``inputs'' of matrices), while filled labels correspond to primal spaces (that is vectors, or ``outputs'' of matrices). Wires connect an empty label with a filled one of the same shape, corresponding to the same vector space. In other words, wires encode tensor contractions $V^* \times V \to \mathbb C$. Presented with a diagram $\mathcal D$ (a collection of boxes and wires) containing boxes associated to a Haar distributed random orthogonal matrix $U \in \mathcal U(n)$, we can interpret the Weingarten formula \eqref{eq:Wg} as a \emph{graph expansion} corresponding to the sum over the pairings $\alpha$ and $\beta$. To each term in the sum we associate a new diagram $\mathcal D_{\alpha,\beta}$ which is obtained by deleting the boxed corresponding to the random matrix $U$, and adding wires encoding the product of delta functions in \eqref{eq:Wg}.  For each pair $(i,j)$ contained in $\alpha$, a wire is added between each primal vector space (i.e.~filled label) of the boxes corresponding to the $i$-th and the $j$-th matrix $U$. Similarly, wires are added between the empty labels, according to the permutation $\beta$. We have thus, assuming $\mathcal D$ contains $2r$ $U$-boxes, 
\begin{equation}\label{eq:Wg-graphical}
\mathbb E_U \mathcal D = \sum_{\alpha, \beta \in \tilde{\mathcal S}_{2r}} \mathcal D_{\alpha,\beta} \operatorname{Wg}_{n}(\alpha, \beta).
\end{equation}

Let us showcase the formula above using a simple example. Let $A \in \mathcal M_n(\mathbb C)$, and let us compute $\mathbb E_U UAU^\top$, for a Haar orthogonal matrix $U \in \mathcal O(n)$. Here, $r=1$, so there is only one possible pairing $\alpha = \beta = (12)$. The original diagram and the graph expansion are represented in Figure \ref{fig:example-Wg}. We conclude that
$$\mathbb E_U UAU^\top =  \operatorname{Tr}(A)I_n \operatorname{Wg}_n((12),(12)) = \frac{1}{n}\operatorname{Tr}(A)I_n.$$

\begin{figure}[htbp]
\includegraphics{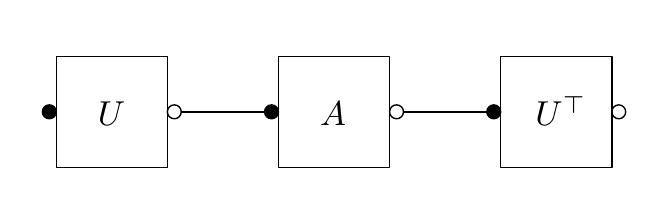} \qquad \includegraphics{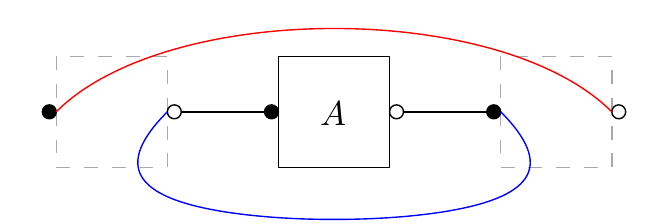}
\caption{On the left, the diagram for the matrix $UAU^\top$. On the right, the only diagram appearing in the graph expansion, obtained by deleting the $U$-boxes, and adding wires corresponding to the $\alpha$-pairing (in red) and to the $\beta$-pairing (in blue).} 
\label{fig:example-Wg}
\end{figure}

\section{Output states for tensor powers of random Haar-orthogonal quantum channels}\label{sec:outputs}
We consider the following model of random quantum channels. We fix an integer $k$ and a real number $t \in (0,1)$, which are the parameters of the model. For each integer $n$, consider the random quantum channel $\Phi_n: \mathcal M_{d_n}(\mathbb C) \to \mathcal M_k(\mathbb C)$, where $d_n := \lfloor tkn \rfloor$ and
\begin{equation}\label{eq:def-Phi_n}
\Phi_n(X) := [\operatorname{id}_k \otimes \operatorname{Tr}_n](V_nXV_n^\top),
\end{equation}
where $V_n : \mathbb R^{d_n} \to \mathbb R^k \otimes \mathbb R^n$ is a Haar distributed random isometry. Note that although $V_n$ is a real matrix, the matrix in \eqref{eq:def-Phi_n} is an element of $\mathcal M_{kn \times d_n}(\mathbb C)$. The random isometry $V_n$ can be obtained by truncating a Haar-distributed random orthogonal matrix $U_n \in \mathcal O(kn)$.

Now we investigate the sequence of random matrices, which are output states of tensor powers of random Haar-orthogonal quantum channels, 
with some fixed sequence of input states. 
More precisely, given a fixed sequence of input states $\rho_n = \psi_n\psi_n^*$, with $\psi_n \in \mathbb C^{rd_n}$, $\|\psi_n\|=1$, let
 $$Z(\rho_n):= \Phi_n^{\otimes r}(\rho_n) \in \mathcal M_{k^r}(\mathbb C).$$
 Our goal in this section will be to characterize the asymptotic behavior of the sequence of random matrices $Z(\rho_n)$. In this setting, the parameters $r,k,t$ are fixed. 
 
The first result is a formula for the moments of the random matrices $Z(\rho_n)$. Let $p \geq 1$ be the order of the moment and we wish to compute $\mathbb E \operatorname{Tr} Z(\rho_n)^p$. We shall use the graphical orthogonal Weingarten formula from Section \ref{sec:Wg-O}. 
We have depicted the diagram for $\operatorname{Tr}Z(\rho_n)^2$, in the case $r=3$, in Figure \ref{fig:trace-p2-r3}. 
 
\begin{figure}[htbp]
\includegraphics{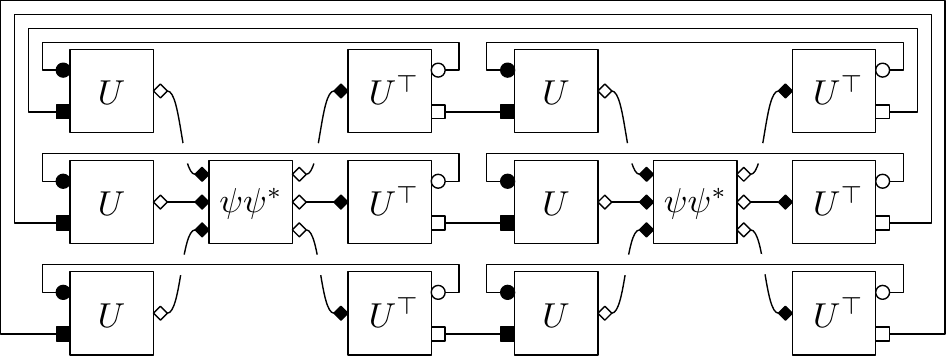}
\caption{A representation of the diagram for the $p=2$ moment of the random matrix $Z(\rho_n)$, in the case where $r=3$ copies of the quantum channel are acting on an input $\rho_n = \psi \psi^*$. Circular decorations correspond to the vector space $\mathbb C^n$, rectangular decorations correspond to $\mathbb C^k$, while diamond-shaped decorations correspond to $\mathbb C^{d_n}$.} 
\label{fig:trace-p2-r3}
\end{figure}

The diagram corresponding to the $p$-th moment contains $p \times r \times 2$ random orthogonal matrices $U \in \mathcal O(kn)$. We shall index these matrices by a triple $[i,x,P]$, where
 \begin{itemize}
\item the label $i \in \{1, \ldots, p\}$ indicates the index of the copy of the matrix $Z(\rho_n)$ the $U$ box belongs to;
\item the label $x \in \{1, \ldots, r\}$ denotes the index of the channel $\Phi_n$ in the tensor power;
\item the position label $P \in \{L,R\}$ indicates whether the box $U$ appears on the ``left'' side of the picture or on the ``right'' side (i.e.~the matrix $U$ appears without or with a transposition in \eqref{eq:def-Phi_n}).
\end{itemize}
We introduce now two permutations which encode the initial wiring (tensor contractions) appearing in the diagram. To this end, we identify the set of integers $\{1, \ldots, 2pr\}$ with the set of triples $[i,x,P]$ described above. We put
\eqb{\label{eq:delta-gamma}
\delta &:= \prod_{i=1}^p \prod_{x=1}^r ([i,x,L], [i,x,R])\\
\gamma &:=  \prod_{i=1}^p \prod_{x=1}^r ([i,x,L], [i-1,x,R]).
}
In the second equation above, we abuse notation and write $[0,x,P]:=[p,x,P]$ for any index $x$ and position $P$. It is important to notice that both permutations above are products of $pr$ disjoint transpositions, so $\delta, \gamma \in \tilde{\mathcal S}_{2pr}$. As we shall see, the permutations $\delta, \gamma$ encode the wirings corresponding to the partial trace (for each quantum channel) and, respectively, to the trace appearing in the moment of $Z(\rho_n)$. 

The graphical formulation of the Weingarten formula for integrals over the orthogonal group $\mathcal O(kn)$ gives
\begin{equation}\label{eq:sum-Wg}
\mathbb E \operatorname{Tr} Z(\rho_n)^p = \sum_{\alpha, \beta \in  \tilde{\mathcal S}_{2pr}} \mathcal D_{\alpha,\beta} \operatorname{Wg}_{kn}(\alpha, \beta),
\end{equation}
where the sum ranges over pairs $(\alpha, \beta)$ of pairings of  the set of $2rp$ boxes containing the random isometry $U$; the permutation $\alpha$ is responsible for pairing the ``outputs'' of the boxes (corresponding to black labels), while $\beta$ pairs the inputs (i.e.~white labels). Let compute explicitly the content of a given diagram $\mathcal D_{\alpha,\beta}$:
\begin{enumerate}
\item Loops corresponding to the partial traces in the quantum channel. Since the original wiring of the boxes corresponding to these loops is encoded by the permutation $\delta$, the contribution of these loops is $n^{\#(\delta\alpha)/2}$, by Lemma \ref{lem:connected-components-pairings}.
\item Loops coming from the matrix multiplication, giving a total contribution of $k^{\#(\gamma\alpha)/2}$ (for the same reasons as above). 
\item The contribution of the input state, let us call it $f_\beta(\rho_n)$ for now. 
\end{enumerate}

Let us bound the contribution of the input state $f_\beta(\rho_n)$. To this end, notice that $f_\beta(\rho_n) = \operatorname{Tr}[(\rho_n)^{\otimes p} M(\beta)]$, where $M(\beta) \in \mathcal M_{d_n}(\mathbb C)^{\otimes pr}$ is a matrix encoding the pairing $\beta$, having $pr$ inputs corresponding to labels $[i,x,L]$ and $pr$ outputs corresponding to labels $[j,y,R]$, see Figure \ref{fig:M-beta} for an example. 

\begin{figure}[htbp]
\includegraphics{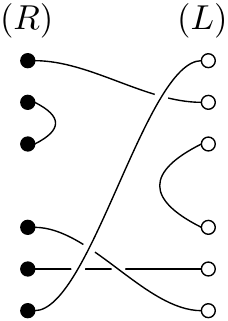}
\caption{An example for the diagram of the matrix $M(\beta)$ encoding the pairing $\beta$ in the case $p=2$, $r=3$.} 
\label{fig:M-beta}
\end{figure}

Let us define, for a pairing $\beta \in \tilde{\mathcal S}_{2q}$ where $q=pr$, its number of \emph{bumps} $\flat(\beta)$ 
as the number of pairs inside $\beta$ which connect elements on the $R$ ``side''. For the pairing $\beta$ in Figure \ref{fig:M-beta}, 
we have $\flat(\beta) = 1$, since there is only  one ``bump'' on the $R$ side. It is obvious that the number of ``bumps'' on the $L$ side is also $\flat(\beta)$, and that, up to multiplying from the left and from the right with some unitary operators, 
the matrix $M(\beta)$ is a tensor product of $\flat(\beta)$ unnormalized maximally entangled states with the identity operator up to rotations. In particular, we have $ \|M(\beta)\|_\infty = d_n^{\flat(\beta)}$, and thus, using H\"older's inequality, we conclude that 
\eq{\label{eq:bound-f}|f_\beta(\rho_n)| \leq d_n^{\flat(\beta)}.}
In order to get a better understanding of the number of bumps of a pairing, let us call a pairing $\tau$ \emph{transverse} if it maps the $L$ side to the $R$ side and vice-versa. In other words, $\tau$ is transverse if for all $(i,x) \in [p] \times [r]$, $\tau([i,x,L]) = [*,*,R]$, and $\tau([i,x,R]) = [*,*,L]$. Note that transverse pairings have zero bumps. We claim the following expression for the number of bumps of a given pairing $\beta$: 

\begin{lemma}\label{lemma:bumps-min}\label{lemma:bumps-minimizer}
For $\beta \in \tilde S_{2q}$
\begin{equation}\label{eq:bumps-min}
2\flat(\beta) = \min_{\tau \text{ transverse}} |\tau \beta|.
\end{equation}
Moreover, the minimum is achieved if and only if $\tau = \tau_1 \oplus  \tau_2$. Here,
$\tau_1 = \prod_{i=1}^{2\flat}(r_i,l_i)$ where for $1\leq j \leq \flat$ each pair 
$\{r_{2j-1},r_{2j}\}$ or $\{l_{2j-1},l_{2j}\}$  supports  a bump in R or L side, respectively,
and $\tau_2 = \prod_{i=2\flat+1}^{2q}(r_i,l_i)$ where $(r_i,l_i) \in \beta$.
\end{lemma} 

\begin{proof}
To prove our claim, we can assume without loss of generality that $2\flat(\beta)=q$, 
i.e., all $2q$ elements are supporting elements of bumps, and $\tau_2 =0$.  
This is because for each transposition $(r,l) \in \beta$ where $r$ and $l$ are from R and L sides, respectively,
we can restrict ourselves to transverse $\tau$ such that $(r,l) \in \tau$ in search for the minimum of $|\tau \beta|$.

To begin with, we prove $\leq$ in \eqref{eq:bumps-min}. 
Consider the bumps on $R$ side and name the supporting elements in pairs 
by $\{r_1,r_2\},\ldots ,\{r_{2\flat -1},r_{2\flat}\}$ where $\flat =\flat(\beta)$.
Then, for a transverse $\tau \in \tilde S_{2q}$ we have the following mapping of $\tau\beta$: for $1\leq j \leq \flat$
\eqb{\label{resolving-transverse-map}
r_{2j-1} &\mapsto l_{2j} \\
r_{2j} &\mapsto l_{2j-1} 
} 
for some distinctive elements $l_1,\ldots,l_{2\flat}$ from $L$ side,
i.e., $\tau(r_i) = l_i$ for $1 \leq i \leq 2\flat$.
Suppose $\tau\beta$ consists of disjoint cycles, say, $c_1, \ldots, c_m$, so that 
\eq{\label{eq:decomposition-tau-gamma}
|\tau\beta| = \sum_{i=1}^m (\mathrm{card}(c_i)-1)
}
where $\mathrm{card}(c_i)$ is the cardinality of cycle $c_i$.
Here, we have $m \leq 2\flat$ based on the comment at the beginning of this proof.
Now, each mapping in \eqref{resolving-transverse-map} constitutes a part of some cycle.
If $c_i$ is related to $k_i$ mappings in \eqref{resolving-transverse-map}, then $\mathrm{card} (c_i)  \geq 2k_i$.
This implies that 
\eq{
|\tau\beta| \geq \sum_{i=1}^m (2 k_i -1) = 4 \flat - m \geq 2 \flat
}
The equality holds if and only if $m=2\flat$ and $\mathrm{card}(c_i) = 2$.
In this case, the condition $\tau \beta (l_{2j}) = r_{2j-1}$ implies that $\beta(l_{2j}) = l_{2j-1}$.
This complets the proof. 
\end{proof}

\begin{lemma}\label{lemma:geodesic-ab}
Given $4m$ elements $\{l_1, \ldots ,l_{2m},r_1,\ldots ,r_{2m}\}$, define two permutations in $\tilde S_{4m}$.
\eqb{
\hat\delta&= \prod_{i=1}^{2m} (r_i,l_i) = \prod_{i=1}^m (r_{2i-1}, l_{2i-1}) (r_{2i}, l_{2i})\\
\hat\beta&= \prod_{i=1}^m (r_{2i-1}, r_{2i}) (l_{2i-1}, l_{2i}).
}
Then, $\hat\alpha \in \tilde S_{4m}$ such that $\hat\delta -\hat \alpha - \hat\beta$ is of the form:
\eq{
\hat\alpha =  \prod_{i \in\Lambda}  (r_{2i-1}, r_{2i}) (l_{2i-1}, l_{2i}) \prod_{i\in [m]\setminus \Lambda} (r_{2i-1}, l_{2i-1}) (r_{2i}, l_{2i})
}
for some $\Lambda \subseteq [m]$. Here we used the notation from \eqref{eq:geodesic-notation}.
\end{lemma}
\begin{proof}
We decompose 
$$\{l_1, \ldots ,l_{2m},r_1,\ldots ,r_{2m}\} = \bigsqcup_{i=1}^m \{l_{2i-1},l_{2i},r_{2i-1},r_{2i}\}$$ and work on each component for the geodesic because $\hat\delta$ and $\hat\beta$ both respect this decomposition.
Note that, at fixed $i$, the only elements on the geodesic are the restrictions of $\hat \delta$ and $\hat \beta$ on the 4-element set:
\eqb{
\mathrm{dist}\big( (r_{2i-1}, r_{2m}) (l_{2i-1}, l_{2m}),(r_{2i-1}, l_{2i-1}) (r_{2m}, l_{2m})\big) =2
}
while the intermediate permutations do not belong to $\tilde S_{4m}$. The proof is now complete, since each block of $\alpha$ must be either of $\hat \delta$ type or of $\hat \beta$ type.
\end{proof}

With these ingredients in hand, and with the asymptotic formula for the Weingarten function from \eqref{eq:Wg-asympt}, 
we can calculate the general term in the sum \eqref{eq:sum-Wg}  and upper bound its absolute value 
as follows (remember our notations of $\delta$ and $\gamma$ in \eqref{eq:delta-gamma}): 
\eq{\label{eq:moment1}
\mathcal D_{\alpha,\beta} \operatorname{Wg}_{kn}(\alpha, \beta) 
&= n^{\#(\delta\alpha)/2} k^{\#(\gamma\alpha)/2} f_\beta(\psi_n)\operatorname{Wg}_{kn}(\alpha,\beta),
\qquad\text{and then}\\
\label{eq:moment2}
|\mathcal D_{\alpha,\beta} \operatorname{Wg}_{kn}(\alpha, \beta)| 
&\leq (1+o(1)) \left[n^{\#(\delta\alpha)/2} k^{\#(\gamma\alpha)/2} (tkn)^{\flat(\beta)} (kn)^{-pr - |\alpha\beta|/2}|\operatorname{\textnormal{M\"ob}}(\alpha,\beta)|\right].
}
By using \eqref{eq:bumps-min} in Lemma \ref{lemma:bumps-min}
the exponent of $n$ (the only variable which grows) in the RHS of \eqref{eq:moment2} reads
\begin{align}
\#(\delta\alpha)/2 +  \min_{\tau \text{ transverse}} |\tau \beta|/2 - pr - |\alpha\beta|/2  
&= (\min_{\tau \text{ transverse}} |\tau \beta| - |\delta\alpha| - |\alpha\beta|)/2\notag\\
&\leq (\min_{\tau \text{ transverse}} |\tau \beta| - |\delta\beta| )/2 \label{eq:bound-a}\\
&\leq 0, \label{eq:bound-b}
\end{align}
where we have used the triangle inequality and the fact that the permutation $\delta$ is transverse. 
To identify the leading order terms in \eqref{eq:sum-Wg}
we then try to ignore as many terms as possible by getting rid of terms which do not saturate the three bounds
\eqref{eq:bound-f}, \eqref{eq:bound-a} and \eqref{eq:bound-b}.
Note that we consider the bound \eqref{eq:bound-f} only asymptotically, as one can see below.

First, the equality $\displaystyle \min_{\tau \text{ transverse}} |\tau \beta|= |\delta\beta| $ must hold in \eqref{eq:bound-b}.
Since $\delta$ is a transverse, Lemma \ref{lemma:bumps-minimizer} shows that 
$\beta$ must be of the form
\begin{align} 
\label{eq:beta-dominant} \beta &= \prod_{B} ([i_1(s),x_1(s),L],[i_2(s),x_2(s),L])([i_1(s),x_1(s),R],[i_2(s),x_2(s),R])\\
\nonumber &\quad \times \prod_{B^c} ([i_3(t), x_3(t),L],[i_3(t),x_3(t),R]).
\end{align}
In other words, $\beta$ must be a product of symmetrical bumps and horizontal wires. 
Here, $B \in \mathcal C_p$, and $\mathcal C_p$ is defined by a set of particular types of transpositions:
\eq{
\mathcal C_p &= \Big\{  \big([i_1(s),x_1(s)], [i_2(s),x_2(s)]\big)  \Big\}_{s=1}^m: \quad m \in \left[\left\lfloor \frac{pr}{2} \right\rfloor \right],\\
&[i_j(s),x_j(s)] \not = [i_l(t),x_l(t)] \text{ unless $j=l$ and $s=t$}\Big\}.
}
Also, we abuse notations by writing $B^c$ to denote fixed points in $[pr]$ by all transpositions in $B$.

Second, the equality in \eqref{eq:bound-a} holds if and only if $\alpha$ lies on the geodesic between $\delta$ and $\beta$. 
This is equivalent via Lemma \ref{lemma:geodesic-ab} to the fact that $\alpha$ has the following form
for $A \in \mathcal C_p$ such that $A \subseteq B$ and:
\begin{align}
\label{eq:alpha-dominant}\alpha &= \prod_{A} ([i_1(s),x_1(s),L],[i_2(s),x_2(s),L])([i_1(s),x_1(s),R],[i_2(s),x_2(s),R])\\
\nonumber &\quad \times \prod_{A^c} ([i_3(t), x_3(t),L],[i_3(t),x_3(t),R]).
\end{align}
In other words, $\alpha$ consists of horizontal lines and a \emph{subset} of the bumps of $\beta$.

Third, we discuss when the equality in \eqref{eq:bound-f} is asymptotically saturated when $p=2$. 
To this end, we define $\flat_{\mathrm{in}}(\beta)$ the number of ``non-trespassing'' bumps for $\beta$ defined in \eqref{eq:beta-dominant}.
For this aim, we define
\eq{\label{eq:non-trespassing-b}
B_{\mathrm{in}} =  \left\{ \big([i_1,x_1], [i_2,x_2]\big) \in B:  i_1 = i_2 \right \}, 
}
where ``$\in$'' means that the left transposition is one of transpositions constituting $B$,
so that we have the definition of $\flat_{\mathrm {in}}(\beta) = \left|B_{\mathrm{in}}\right|$.

Now, we need a lemma:
\begin{lemma}\label{lemma:non-trespassing}
For $\beta$ defined in \eqref{eq:beta-dominant}, we have the following bound for $p=2$.
\eq{
|f_\beta(\rho_n)| \leq d_n^{\flat_{\mathrm{in}}(\beta)}
}
\end{lemma}
\begin{proof}
Let $\omega_C$ be a maximally entangled state associated to $C \in \mathcal C_p$,
i.e.~a tensor product of maximally entangled states, each of which is defined by a transposition in $C$ (see Section \ref{sec:qit} for the definitions).
Then, using the general ``linearization trick''
$$\operatorname{Tr}(XY^T) = \operatorname{Tr}[\omega (X \otimes Y) \omega],$$
we get
\eqb{
f_\beta (\rho_n)& = d_n^{\flat(\beta)}\cdot \trace_{B^c} \left[ \left(\hat\omega_{B_{\mathrm{in}}}^* \otimes \hat\omega_{B \setminus B_{\mathrm{in}}}^*\otimes I_{B^c} \right)
\rho_n\otimes \rho_n\left( \hat\omega_{B_{\mathrm{in}}}\otimes \hat\omega_{B \setminus B_{\mathrm{in}}} \otimes I_{B^c} \right) \right]\\
&\leq d_n^{\flat(\beta)}  \trace_{B_{\mathrm{in}} \otimes B^c} \left[ \left(I_{B_{\mathrm{in}}}\otimes \hat\omega_{B \setminus B_{\mathrm{in}}}^*\otimes I_{B^c} \right)
\rho_n\otimes \rho_n\left( I_{B_{\mathrm{in}}}\otimes \hat\omega_{B \setminus B_{\mathrm{in}}} \otimes I_{B^c} \right) \right]\\
&= d_n^{\flat_{\mathrm{in}}(\beta)}  \trace \left[\Psi^{(1)}\Psi^{(2)T} \right] \leq d_n^{\flat_{\mathrm{in}}(\beta)} 
}
where $\Psi^{(1)}$ and $\Psi^{(2)}$ are reduced density operators of $\rho_n$ in the first and second spaces, and we have used the trivial matrix inequality $\hat \omega \leq I$.
\end{proof}
This means that we can reduce candidates of leading order terms in \eqref{eq:sum-Wg}, and for writing purpose
we define the set of non-trespassing bumps by
\eq{
\mathcal C_{p,\mathrm{in}} = \{B \in \mathcal C_p: B_{\mathrm{in}} = B \}
}
Note that trivially $\mathcal C_1 = \mathcal C_{1, \mathrm{in}}$.
Then, finally, we can state the result giving the asymptotic moments of the sequence of random matrices $Z(\psi_n)$. 
From here on, we  identify $\alpha, \beta$ with $A,B \in \mathcal C_p$.
\begin{theorem}\label{thm:moments}
For any given sequence of input states $\rho_n$, \\
1) All moments of $Z(\rho_n)$ are expressed as
\eq{
 (1+o(1))  \sum_{\substack{B\in\mathcal C_p\\A \subseteq B}} 
k^{\frac{\#(\gamma \alpha)}{2} +|A|-pr}\cdot t^{|B|}\cdot g_B(\rho_n) \cdot (-1)^{|B|-|A|}
}
where 
\eq{
g_B(\rho_n)=  \frac{f_\beta(\rho_n)}{(tnk)^{|B|}} \leq 1
}
\\
2) For the first and second moments of $Z(\rho_n)$ one can replace $\mathcal C_p$ by $\mathcal C_{p,\mathrm{in}}$.

\end{theorem}
\begin{proof}
For pairings $\alpha$ and $\beta$ as in \eqref{eq:alpha-dominant}, resp.~\eqref{eq:beta-dominant}, the M\"obius function is given by \eqref{eq:Mob}:
$$\operatorname{\textnormal{M\"ob}}(\alpha,\beta) = (-1)^{| B \setminus A|} = (-1)^{| B|-|A|} .$$
Also note that $\flat(\beta) = |B|$ for $\beta$ in \eqref{eq:beta-dominant}.
Neglecting terms in \eqref{eq:moment1} which vanish according to the above discussions, the general moment an be written, except for the $(1+o(1))$ factor, as 
\eqb{\label{eq:moment-g}
\sum_{\alpha, \beta \text{ as in } \eqref{eq:alpha-dominant}, \eqref{eq:beta-dominant}} 
& k^{\#(\gamma \alpha) /2}\cdot (tk)^{|B|} \cdot g_B(\rho_n)\cdot  k^{-pr-|\alpha\beta|/2}\mob(\alpha,\beta)\\
&= \sum_{\substack{B\in\mathcal C_p\\A \subseteq B}} 
k^{\frac{\#(\gamma \alpha)}{2} +|A|-pr}\cdot t^{|B|}\cdot g_B(\rho_n) \cdot (-1)^{|B|-|A|}
}
which is the general formula we wanted.
Moreover we can replace $\mathcal C_p$ by $\mathcal C_{p,\mathrm{in}}$ for $p=1,2$,
based on Lemma \ref{lemma:non-trespassing} and the remark following it.
\end{proof}

Next, we calculate the average output state for a fixed input $\rho_n$.
To this end, we introduce a useful notation before going onto our theorem. 
Define for $A \in \mathcal C_1$
\eq{
T_A^{(k)} := \left[ \bigotimes_{\{i,j\} \in A} \omega_{ij} \right] \otimes \left[ \bigotimes_{s \notin A} I_s \right]
}
where we denote by $\omega$ the (un-normalized) maximally entangled state $\omega = \Omega \Omega^*$ with
$\Omega = \sum_{i = 1}^k e_i \otimes e_i \in \mathbb C^k \otimes \mathbb C^k$, see also Section \ref{sec:qit}. 
We write $\omega_{ij}$ for the operator $\omega$ acting on the copies $i$ and $j$ of the space $\mathbb C^k$. 
We also abuse notation so that  $s \notin A$ means that $s \in [r]$ stays fixed by transpositions in $A \in \mathcal C_1$.
Then, 
\begin{theorem}\label{theorem:average-output}
\eq{
\label{eq:E-Z-psi}
 \mathbb E Z(\rho_n)=   (1+o(1)) M(\rho_n)
}
where
\eq{\label{eq:M-matrix}
M(\rho_n) := \sum_{\substack{B\in\mathcal C_1\\A \subseteq B}} 
T_A^{(k)} \cdot k^{|A|-r }\cdot t^{|B|}\cdot g_{B}(\rho_n) \cdot (-1)^{|B|-|A|}.
}
\end{theorem}
\begin{proof}
Now we calculate ``the first moment without trace''. 
To this end,
we just replace $k^{\#(\gamma\alpha)/2}$ in \eqref{eq:moment-g} by $T_A^{(k)}$.
In fact $\trace T_A^{(k)} = k^{\frac{\#(\alpha)}{2}}=k^{|A|}$ where $\gamma=\delta$ for $p=1$.
\end{proof}

\begin{theorem}\label{thm:convergence-in-prob}
For a fixed sequence of input states $(\rho_n)_{n \geq 1}$ we have the following convergence in probability:
\eq{
\left\| Z(\rho_n) -  \mathbb E Z(\rho_n) \right\|_2 \to 0
}
\end{theorem}
\begin{proof}
Using the second part of Theorem \ref{thm:moments}, the second moment of $Z(\rho_n)$ is a sum indexed by sets $B \in \mathcal C_{2,\mathrm{in}}$. For such a $B$,  we write $B = B_1 \oplus B_2$ where these two belong to blocks with $i=1,2$ respectively, so that, using the notation from Theorem \ref{thm:moments}, we can factorize
\eq{
g_B (\rho_n) = g_{B_1}(\rho_n) \cdot g_{B_2}(\rho_n)
}
Then, 
the formula in \eqref{eq:moment-g} with $p=2$, which represents the second moment, up to $o(1)$ terms, changes into:
\eqb{
& \sum_{\substack{B_1 \oplus B_2\in\mathcal C_{2,\mathrm{in}}\\A_1 \oplus A_2 \subseteq B_1 \oplus B_2}} 
k^{\frac{\#(\gamma (\alpha_1 \oplus \alpha_2))}{2} +|A_1| +|A_2|-2r}\cdot t^{|B_1|+|B_2|}\cdot g_{B_1}(\rho_n) \cdot g_{B_2}(\rho_n) \cdot (-1)^{|B_1|+|B_2|-|A_1| -|A_2|}\\
&=\trace  \left[ \prod_{i=1}^2 \left( \sum_{\substack{B_i\in\mathcal C_1\\A_i \subseteq B_i}} 
T_{A_i}^{(k)} \cdot k^{|A_i|-r}\cdot t^{|B_i|}\cdot g_{B_i}(\rho_n) \cdot (-1)^{|B_i|-|A_i|} \right) \right]
= \trace \left[ (M(\rho_n) )^2 \right] + o(1),
}
where $\alpha_i$ are defined by $A_i$, respectively.
Then, Chebyshev's inequality shows for each $\varepsilon >0$
\begin{align*}
\mathbb{P} \left( \left\| Z(\rho_n)- \mathbb E Z(\rho_n)  \right\|_2^2 \geq \varepsilon^2 \right)
&\leq \frac{1}{\varepsilon^2} \mathbb E  \left\| Z(\rho_n) - \mathbb E Z(\rho_n)  \right\|_2^2 \\
&= \frac{[\mathbb E \operatorname{Tr}Z(\rho_n)]^2 - \operatorname{Tr}[M(\rho_n)^2] + o(1)}{\varepsilon^2} =  \frac{o(1)}{\varepsilon^2}
\end{align*}
This completes our proof of the convergence in probability.
\end{proof}

\begin{remark}
For some models of random unitary channels, it is possible to show that similar convergence results hold \emph{almost surely}, a stronger convergence that the convergence in probability proven here. This is enabled by better controlling the error in equations such as \eqref{eq:E-Z-psi}, up to $O(n^{-2})$ terms. This is one technical difference between random unitary and random orthogonal  matrices: in the former case, the error in the approximation of the Weingarten formula \eqref{eq:Wg-asympt} is $O(n^{-2})$, while in the latter it is $O(n^{-1})$, see \cite{collins2006integration}.

\end{remark}

\section{Optimal sequences of input states}\label{sec:optimal-inputs}

Having computed in the previous section the asymptotic behavior of the outputs for a fixed sequence of input state, we turn now to the problem of finding the input sequences giving the outputs with least entropy (asymptotically). Our strategy is to show that for any sequence of input states, the outputs will lie, asymptotically, inside a fixed, deterministic set $K_{r,k,t}$. We shall then minimize the entropy for states inside this convex set $K_{r,k,t}$. 

We start by writing the expected value of an output state into a more compact form. 
In what follows we replace $\mathcal C_1$ by $\hat{\mathcal P}_2(r)$ the set of partial parings on $[r]$
because in this section the parameter $r$ is more relevant. 
Starting from $M(\rho_n)$ in \eqref{eq:M-matrix}, we have
\begin{align}
\nonumber M(\rho_n) &= \sum_{A \subseteq B \in \hat{\mathcal P}_2(r)} T_A^{(k)} t^{|B|} k^{-r+|A|} g_B(\rho_n) (-1)^{|B|-|A|}\\
\nonumber &=\sum_{ B \in \hat{\mathcal P}_2(r)} \langle  \tilde T_B^{(d_n)} , \rho_n \rangle \sum_{A \subseteq  B}  t^{|B|} k^{-r+|A|} (-1)^{|B|-|A|} T_A^{(k)}\\
\label{eq:Z-psi-T-R}&= \sum_{ B \in \hat{\mathcal P}_2(r)} \langle \tilde T_B^{(d_n)}, \rho_n \rangle \tilde R_B^{(k)},
\end{align}
where the operators $ \tilde T_B^{(d_n)} \in \mathcal M_{d_n^r}(\mathbb C)$ and  $\tilde R_B^{(k)} \in \mathcal M_{k^r}(\mathbb C)$  for $A,B \in \hat{\mathcal P}_2(r)$ are defined as follows:
\begin{align*} 
\tilde T_B^{(d_n)} &:= d_n^{-|B|} T_B^{(d_n)} \quad \left(=  \left[ \bigotimes_{\{i,j\} \in B} d_n^{-1}\omega_{ij} \right] \otimes \left[ \bigotimes_{s \notin B}  I_s \right] \right)\\
\tilde R_B^{(k)} &:= \left[ \bigotimes_{\{i,j\} \in B} t\left(k^{-1}\omega_{ij}  - k^{-2}I_{ij}\right)\right] \otimes \left[ \bigotimes_{s \notin B} k^{-1} I_s \right]\\
&= \sum_{A \subseteq B}  t^{|B|} k^{-r+|A|} (-1)^{|B|-|A|} T_A^{(k)}
\end{align*} 
where one can see the last equality via binomial formula.

Note that equation \eqref{eq:Z-psi-T-R} is close to what we want: to express the output of the channel as a convex combination of simple quantum states. The problem here is that, although the scalars $\langle \tilde T_B^{(n)} , \rho_n \rangle$ are non-negative, the matrices $\tilde R_B^{(k)}$ are not, in general, positive semidefinite. In fact, we have $\operatorname{Tr} \tilde R_B^{(k)} = \delta_{B, \emptyset}$. In order to achieve our goal, we shall apply the M\"obius inversion formula \cite{rota1964foundations} to \eqref{eq:Z-psi-T-R}. First, it is quite obvious to see that the M\"obius function on the lattice $\hat{\mathcal P}_2(r)$ is identical to the one for the lattice of subsets: if a partial pairing $A$ is contained in another partial pairing $B$, then $\mu(A, B) = (-1)^{|B|-|A|}$. Hence, if we define
\begin{align}
\label{eq:S-R}\tilde S^{(k)}_B &:= \sum_{A \subseteq B} \tilde R_A^{(k)}\\ 
\nonumber \tilde Q_A^{(d_n)} &:= \sum_{B \supseteq A}  (-1)^{|B|-|A|}\tilde T_B^{(d_n)},
\end{align}
we have, via the M\"obius inversion formula
$$\tilde R^{(k)}_B = \sum_{A \subseteq B} (-1)^{|B|-|A|}\tilde S_A^{(k)},$$
and we can rewrite \eqref{eq:Z-psi-T-R} as 
\begin{align}
\nonumber M(\rho_n)&=  \sum_{ B \in \hat{\mathcal P}_2(r)} \langle \tilde T_B^{(d_n)} , \rho_n \rangle \tilde R_B^{(k)}\\
\nonumber &=  \sum_{ A \subseteq  B \in \hat{\mathcal P}_2(r)} \langle  \tilde T_B^{(d_n)} , \rho_n \rangle (-1)^{|B|-|A|}\tilde S_A^{(k)}\\
\nonumber &=  \sum_{ A  \in \hat{\mathcal P}_2(r)} \left\langle \sum_{B \supseteq  A}(-1)^{|B|-|A|} \tilde T_B^{(d_n)} , \rho_n \right\rangle \tilde S_A^{(k)}\\
\label{eq:Z-psi-Q-S}&=  \sum_{ A  \in \hat{\mathcal P}_2(r)} \langle  \tilde Q_A^{(d_n)} , \rho_n \rangle \tilde S_A^{(k)}.
\end{align}
From \eqref{eq:S-R}, we can actually obtain an explicit formula for the matrices $\tilde S_B^{(k)}$:
\begin{align}
\nonumber\tilde S^{(k)}_B &:= \sum_{A \subseteq B} \tilde R_A^{(k)}\\ 
\nonumber&= \sum_{A \leq B} \left[ \bigotimes_{\{i,j\} \in A} t\left(k^{-1}\omega_{ij}  - k^{-2}I_{ij}\right)\right] \otimes \left[ \bigotimes_{s \notin A} k^{-1} I_s \right] \\
\nonumber&=  \left[ \bigotimes_{\{i,j\} \in B} t\left(k^{-1}\omega_{ij}  - k^{-2}I_{ij}\right) + k^{-2}I_{ij}\right] \otimes \left[ \bigotimes_{s \notin B} k^{-1} I_s \right] \\
\nonumber&=  \left[ \bigotimes_{\{i,j\} \in B} tk^{-1}\omega_{ij}+(1-t) k^{-2}I_{ij}\right] \otimes \left[ \bigotimes_{s \notin B} k^{-1} I_s \right] \\
\label{eq:S-explicit}&=  \left[ \bigotimes_{\{i,j\} \in B}\eta_{ij}\right] \otimes \left[ \bigotimes_{s \notin B} k^{-1} I_s \right],
\end{align}
where 
\begin{equation}\label{eq:def-eta}
\eta := tk^{-1}\omega+(1-t) k^{-2}I \in \mathcal M_{k^2}(\mathbb C)
\end{equation}
is indeed a quantum state (i.e.~a positive semidefinite matrix of unit trace); such states, convex mixtures between a maximally entangled state and a maximally mixed state are called \emph{isotropic states} in the quantum information theory literature. 

We have now all the ingredients to state the main result of this section. 

\begin{theorem}\label{thm:optimal-input-seq}
Consider a sequence of random quantum channels $\Phi_n : \mathcal M_{d_n}(\mathbb C) \to \mathcal M_k(\mathbb C)$ constructed from random Haar distributed orthogonal matrices $U_n \in \mathcal O(kn)$, as in Section \ref{sec:outputs}. Furthermore, assume that $d_n \sim tkn$ for some constant $t \in (0,1)$ and define, for any $r \geq 1$, the convex set
$$K_{r,k,t}:= \operatorname{conv} \left\{ \tilde S^{(k)}_B \, : \, B \in \hat{\mathcal P}_2(r) \right \}\subseteq \mathcal M_{k^r}^{1,+}(\mathbb C).$$
Then, for any \emph{fixed} sequence of input states $\rho_n \in  \mathcal M_{d_n}^{1,+}(\mathbb C)$, the output states converge, in probability, to the convex body $K_{r,k,t}$: for all $\varepsilon >0$,
$$\lim_{n \to \infty} \mathbb P\left[ \operatorname{dist}(\Phi_n^{\otimes r}(\rho_n), K_{r,k,t}) > \varepsilon \right] = 0.$$
Note that $K_{r,k,t}$ depends on $t$ via \eqref{eq:def-eta}.
\end{theorem}
\begin{proof}
Let us fix a sequence of input states $(\rho_n)$ and use the triangle inequality:
$$\operatorname{dist}(\Phi_n^{\otimes r}(\rho_n), K_{r,k,t}) \leq \operatorname{dist}(\mathbb E\Phi_n^{\otimes r}(\rho_n), K_{r,k,t})  + \|\Phi_n^{\otimes r}(\rho_n) -\mathbb E \Phi_n^{\otimes r}(\rho_n)\|_2.$$
We have shown in Theorem \ref{thm:convergence-in-prob} that the second term in the right hand side of the above inequality converges in probability towards zero; it is enough thus to show that the first term also vanishes as $n \to \infty$. From \eqref{eq:Z-psi-Q-S}, we have the following decomposition
$$\mathbb E\Phi_n^{\otimes r}(\rho_n) = (1+o(1))\sum_{ A  \in \hat{\mathcal P}_2(r)} \langle \tilde Q_A^{(d_n)} , \rho_n \rangle \tilde S_A^{(k)}.$$
To finish the proof, we show next that the weights in the equation above are (asymptotically) non-negative and sum up to one. For the claim about the sum, note that
$$\sum_{ A  \in \hat{\mathcal P}_2(r)} \tilde Q_A^{(d_n)}=\sum_{ A \subseteq B  \in \hat{\mathcal P}_2(r)} (-1)^{|B|-|A|} \tilde R_A^{(d_n)}  = \tilde T_\emptyset^{(d_n)} = I_{k^r},$$
proving the claim. The other claim follows from \cite[Corollary 3.6]{fukuda2014asymptotically}, where it was shown that the spectrum of the matrices $\tilde Q_A^{(d_n)}$ is at distance $O(1/n)$ from the set $\{0,1\}$. The reader should make note of the fact that although the matrices $\tilde Q^{(d_n)}_\cdot$ are indexed by different combinatorial objects (partial pairings here and partial permutations in \cite{fukuda2014asymptotically}), they encode the same linear operators and thus they have the same spectrum. 
\end{proof}

\begin{corollary}\label{cor:min-entropy}
Let $B_0$ be a maximal partial pairing in $\hat{\mathcal P}	_2(r)$, i.e.~a pairing consisting of $\lfloor r/2 \rfloor$ pairs and, when $r$ is odd, a singleton. Then, for any \emph{fixed} sequence of input states $\rho_n \in  \mathcal M_{d_n}^{1,+}(\mathbb C)$, the inputs 
$$G^{(d_n)}_{B_0}:= \left[ \bigotimes_{\{i,j\} \in B_0} d_n^{-1}\omega_{ij} \right] \otimes \left[ \bigotimes_{s \notin B_0} d_n^{-1} I_s \right] = d_n^{2\lfloor r/2 \rfloor-r} \tilde T^{(d_n)}_{B_0} $$
 give output states having less entropy than the sequence of inputs $\rho_n$: for all $\varepsilon >0$,
$$\lim_{n \to \infty} \mathbb P\left[ H\left( \Phi_n^{\otimes r}(\rho_n)\right) <  H\left(\Phi_n^{\otimes r}(G^{(d_n)}_{B_0})\right)  - \varepsilon \right] = 0.$$
In other words, the sequence of input states consisting of a tensor product of $\lfloor r/2 \rfloor$ maximally entangled states and, when $r$ is odd, a maximally mixed state yields the output sequence with least asymptotical entropy.  
\end{corollary}
\begin{proof}
By the theorem, the outputs belong, when $n$ is large, to the set $K_{r,k,t}$. The extremal points of $K_{r,k,t}$ are precisely the quantum states $\tilde S^{(k)}_B$, with $B$ a partial pairing of $[r]$. Such an extremal state has von Neumann entropy
$$H(\tilde S^{(k)}_B) = |B| H(\eta) + (r-2|B|) \log k,$$
where $\eta$ is the bipartite quantum state define in \eqref{eq:def-eta}; it has entropy strictly less than $2\log k$, more precisely
$$H(\eta) = h(tk^{-1} + (1-t)k^{-2}) + (k^2-1)h((1-t)k^{-2}) ,$$
where $h(x) = -x \log x$.
To finish the proof, we show that the input sequence $G^{(d_n)}_{B_0}$ produces the output sequence $\tilde S^{(k)}_{B_0}$. Indeed, from \eqref{eq:Z-psi-Q-S}, we have
\begin{align*}
\mathbb E \Phi_n^{\otimes r}(G^{(d_n)}_{B_0}) &=  (1+o(1))\sum_{ A  \in \hat{\mathcal P}_2(r)} \langle  \tilde Q_A^{(d_n)} , G^{(d_n)}_{B_0} \rangle \tilde S_A^{(k)} \\
&=  (1+o(1))\sum_{ A \subseteq B  \in \hat{\mathcal P}_2(r)}(-1)^{|B|-|A|}d_n^{2\lfloor r/2 \rfloor-r} \langle  \tilde T_B^{(d_n)} , \tilde T^{(d_n)}_{B_0} \rangle \tilde S_A^{(k)}.
\end{align*}
By direct inspection, and using the fact that $B_0$ is a maximal partial pair pairing, we have that (see also \cite[Section 3]{fukuda2014asymptotically})
$$d_n^{2\lfloor r/2 \rfloor-r} \langle  \tilde T_B^{(d_n)} , \tilde T^{(d_n)}_{B_0} \rangle = (1+o(1)) \mathbf{1}_{B \subseteq B_0},$$
and thus
$$\mathbb E \Phi_n^{\otimes r}(G^{(d_n)}_{B_0}) =  (1+o(1))\sum_{ A \subseteq B \subseteq B_0 \in \hat{\mathcal P}_2(r)}(-1)^{|B|-|A|} \tilde S_A^{(k)} = \tilde S_{B_0}^{(k)},$$
finishing the proof.
\end{proof}

\section{Discussion}\label{sec:discussion}

In this work, using Weingarten calculus on the orthogonal group, we have shown that among fixed input sequences for a tensor power of a random orthogonal quantum channel, product of maximally entangled states achieve the smallest output entropy. 
We consider our results to be evidence toward the claim that such random channels do not violate (asymptotically, with high probability) the additivity relation. 
More precisely, for $r \geq 1$ we conjecture that, almost surely for random orthogonal quantum channels such as the ones in Section \ref{sec:outputs} 
\eq{
\lim_{n \to \infty} S_{\min}(\Phi_n^{\otimes 2r}) \stackrel{?}{=} r \lim_{n \to \infty} S_{\min}(\Phi_n^{\otimes 2}).
}

For this conjecture we must refer to a sentence in \cite{hastings2009superadditivity}: 
``This two-letter additivity conjecture would enable us to restrict our attention to considering input states with a bipartite entanglement structure, 
possibly opening the way to computing the capacity for arbitrary channels''. 
Hastings conjectures thus the following additivity for quantum channels:
\eq{
  S_{\min}((\Psi \otimes \bar \Psi)^{\otimes r})  \stackrel{?}{=}  r  S_{\min}(\Psi \otimes \bar \Psi)
}
In \cite{fukuda2014asymptotically}, we have studied this question in the frame work of the current work, but with random unitary quantum channels.
Then, we have shown that among a very large class of fixed input sequences, tensor products of maximally entangled states yield the outputs with least entropy. 
This is a strong supporting mathematical evidence towards Hastings' conjecture. In the same direction, see \cite{montanaro2013weak,fukuda2015additivity} for considerations about upper bounds on the amount of additivity violations for random quantum channels. 

Surprisingly, if we compare our calculations with ones for unitary random quantum channels from \cite{collins2012towards},
we are inclined to conjecture that generically entanglement does not help to improve minimum output entropy of tensor powers of random unitary quantum channels,
while (only) bipartite entanglement helps for random orthogonal channels: almost surely,
\eq{
\lim_{n \to \infty}  \lim_{r \to \infty}\frac 1rS_{\min}(\Psi_n^{\otimes r}) \stackrel{?}{=} \lim_{n \to \infty} S_{\min}(\Psi_n)  
\quad \text{and}\quad
 \lim_{n \to \infty}\lim_{r \to \infty}\frac 1r S_{\min}(\Phi_n^{\otimes r})\stackrel{?}{=}\frac 12 \lim_{n \to \infty} S_{\min}(\Phi_n^{\otimes 2})  
}
where $\Psi_n$ and $\Phi_n$ are sequences of respectively unitary and orthogonal random quantum channels.

We also conjecture that similar phenomena might occur for the Holevo capacity too, and we hope that such results might shed light on capacity formulas. 
Indeed, according to \cite{collins2015convergence}, certain random quantum channels satisfy a simple linear relation between their Holevo capacity and their minimum output entropy,
while such a linear relation was initially observed in \cite{holevo2005a} for covariant channels.

\bibliography{ref}{}

\begin{thebibliography}{KMNR07}

\bibitem[ASW11]{aubrun2011hastingss}
Guillaume Aubrun, Stanis{\l}aw Szarek, and Elisabeth Werner.
\newblock Hastings' additivity counterexample via {D}voretzky's theorem.
\newblock {\em Communications in mathematical physics}, 305(1):85--97, 2011.

\bibitem[Ban10]{banica2010orthogonal}
Teodor Banica.
\newblock The orthogonal weingarten formula in compact form.
\newblock {\em Letters in Mathematical Physics}, 91(2):105--118, 2010.

\bibitem[BB09]{bjelakovic2009classical}
Igor Bjelakovic and Holger Boche.
\newblock Classical capacities of compound and averaged quantum channels.
\newblock {\em IEEE Transactions on Information theory}, 55(7):3360--3374,
  2009.

\bibitem[BCN12]{belinschi2012eigenvectors}
Serban Belinschi, Beno{\^\i}t Collins, and Ion Nechita.
\newblock Eigenvectors and eigenvalues in a random subspace of a tensor
  product.
\newblock {\em Inventiones mathematicae}, 190(3):647--697, 2012.

\bibitem[BCN16]{belinschi2016almost}
Serban~T Belinschi, Benoit Collins, and Ion Nechita.
\newblock Almost one bit violation for the additivity of the minimum output
  entropy.
\newblock {\em Communications in Mathematical Physics}, 341(3):885--909, 2016.

\bibitem[CFN12]{collins2012towards}
Beno{\^\i}t Collins, Motohisa Fukuda, and Ion Nechita.
\newblock Towards a state minimizing the output entropy of a tensor product of
  random quantum channels.
\newblock {\em Journal of Mathematical Physics}, 53(3):032203, 2012.

\bibitem[CFN15]{collins2015convergence}
Benoit Collins, Motohisa Fukuda, and Ion Nechita.
\newblock On the convergence of output sets of quantum channels.
\newblock {\em Journal of Operator Theory}, 73(2):333--360, 2015.

\bibitem[CM09]{collins2009some}
Beno{\^\i}t Collins and Sho Matsumoto.
\newblock On some properties of orthogonal weingarten functions.
\newblock {\em Journal of Mathematical Physics}, 50(11):113516, 2009.

\bibitem[CN10]{collins2010random}
Beno{\^\i}t Collins and Ion Nechita.
\newblock Random quantum channels {I}: graphical calculus and the {B}ell state
  phenomenon.
\newblock {\em Communications in Mathematical Physics}, 297(2):345--370, 2010.

\bibitem[CN11]{collins2011gaussianization}
Beno{\^\i}t Collins and Ion Nechita.
\newblock Gaussianization and eigenvalue statistics for random quantum channels
  (iii).
\newblock {\em The Annals of Applied Probability}, pages 1136--1179, 2011.

\bibitem[CN16]{collins2016random}
Benoit Collins and Ion Nechita.
\newblock Random matrix techniques in quantum information theory.
\newblock {\em Journal of Mathematical Physics}, 57(1), 2016.

\bibitem[Col16]{collins2016haagerup}
Benoit Collins.
\newblock Haagerup's inequality and additivity violation of the minimum output
  entropy.
\newblock {\em arXiv preprint arXiv:1603.00577}, 2016.

\bibitem[C{\'S}06]{collins2006integration}
Beno{\^\i}t Collins and Piotr {\'S}niady.
\newblock Integration with respect to the haar measure on unitary, orthogonal
  and symplectic group.
\newblock {\em Communications in Mathematical Physics}, 264(3):773--795, 2006.

\bibitem[DD07]{datta2007coding}
Nilanjana Datta and Tony~C Dorlas.
\newblock The coding theorem for a class of quantum channels with long-term
  memory.
\newblock {\em Journal of Physics A: Mathematical and Theoretical},
  40(28):8147, 2007.

\bibitem[FK10]{fukuda2010entanglement}
Motohisa Fukuda and Christopher King.
\newblock Entanglement of random subspaces via the {Hastings} bound.
\newblock {\em Journal of Mathematical Physics}, 51(4):042201, 2010.

\bibitem[FKM10]{fukuda2010comments}
Motohisa Fukuda, Christopher King, and David~K Moser.
\newblock Comments on {Hastings}' additivity counterexamples.
\newblock {\em Communications in Mathematical Physics}, 296(1):111--143, 2010.

\bibitem[FN14]{fukuda2014asymptotically}
Motohisa Fukuda and Ion Nechita.
\newblock Asymptotically well-behaved input states do not violate additivity
  for conjugate pairs of random quantum channels.
\newblock {\em Communications in Mathematical Physics}, 328(3):995--1021, 2014.

\bibitem[FN15]{fukuda2015additivity}
Motohisa Fukuda and Ion Nechita.
\newblock Additivity rates and ppt property for random quantum channels.
\newblock {\em Annales math\'ematiques Blaise Pascal}, 22:1--72, 2015.

\bibitem[Fuk14]{fukuda2014revisiting}
Motohisa Fukuda.
\newblock Revisiting additivity violation of quantum channels.
\newblock {\em Communications in mathematical physics}, 332(2):713--728, 2014.

\bibitem[FW07]{fukuda2007simplifying}
Motohisa Fukuda and Michael~M Wolf.
\newblock Simplifying additivity problems using direct sum constructions.
\newblock {\em Journal of mathematical physics}, 48(7):072101, 2007.

\bibitem[GHP10]{grudka2010constructive}
Andrzej Grudka, Micha{\l} Horodecki, and {\L}ukasz Pankowski.
\newblock Constructive counterexamples to the additivity of the minimum output
  r{\'e}nyi entropy of quantum channels for all $p> 2$.
\newblock {\em Journal of Physics A: Mathematical and Theoretical},
  43(42):425304, 2010.

\bibitem[Has09]{hastings2009superadditivity}
Matthew~B Hastings.
\newblock Superadditivity of communication capacity using entangled inputs.
\newblock {\em Nature Physics}, 5(4):255--257, 2009.

\bibitem[Hol98]{holevo1998capacity}
Alexander~S Holevo.
\newblock The capacity of quantum channel with general signal states.
\newblock {\em IEEE Trans. Inform. Theory}, 44(1):269Ã¢Â€Â“ 273, 1998.

\bibitem[Hol05]{holevo2005a}
A.~S. Holevo.
\newblock Additivity conjecture and covariant channels.
\newblock {\em International Journal of Quantum Information}, 03(01):41--47,
  2005.

\bibitem[Kin02]{king2002additivity}
Christopher King.
\newblock Additivity for unital qubit channels.
\newblock {\em Journal of Mathematical Physics}, 43(10):4641, 2002.

\bibitem[Kin03a]{king2003capacity}
Christopher King.
\newblock The capacity of the quantum depolarizing channel.
\newblock {\em IEEE Transactions on Information Theory}, 49(1):221--229, 2003.

\bibitem[Kin03b]{King2003e}
Christopher King.
\newblock Maximal {$p$}-norms of entanglement breaking channels.
\newblock {\em Quantum Inf. Comput.}, 3(2):186--190, 2003.

\bibitem[KMNR07]{king2005properties}
Christopher King, Keiji Matsumoto, Michael Nathanson, and Mary~Beth Ruskai.
\newblock Properties of conjugate channels with applications to additivity and
  multiplicativity.
\newblock {\em Markov Processes and Related Fields}, 13(2):391--423, 2007.

\bibitem[KR01]{king2001minimal}
Christopher King and Mary~Beth Ruskai.
\newblock Minimal entropy of states emerging from noisy quantum channels.
\newblock {\em IEEE Transactions on information theory}, 47(1):192--209, 2001.

\bibitem[Mon13]{montanaro2013weak}
Ashley Montanaro.
\newblock Weak multiplicativity for random quantum channels.
\newblock {\em Communications in Mathematical Physics}, 319(2):535--555, 2013.

\bibitem[Mos15]{mosonyi2015coding}
Mil{\'a}n Mosonyi.
\newblock Coding theorems for compound problems via quantum r{\'e}nyi
  divergences.
\newblock {\em IEEE Transactions on Information Theory}, 61(6):2997--3012,
  2015.

\bibitem[NC10]{nielsen2010quantum}
Michael~A Nielsen and Isaac~L Chuang.
\newblock {\em Quantum computation and quantum information}.
\newblock Cambridge university press, 2010.

\bibitem[NS06]{NicaSpeicher}
Alexandru Nica and Roland Speicher.
\newblock {\em Lectures on the combinatorics of free probability}, volume 335
  of {\em London Mathematical Society Lecture Note Series}.
\newblock Cambridge University Press, Cambridge, 2006.

\bibitem[Rot64]{rota1964foundations}
Gian-Carlo Rota.
\newblock On the foundations of combinatorial theory i. theory of m{\"o}bius
  functions.
\newblock {\em Probability theory and related fields}, 2(4):340--368, 1964.

\bibitem[Sho02]{shor2002additivity}
Peter~W Shor.
\newblock Additivity of the classical capacity of entanglement-breaking quantum
  channels.
\newblock {\em Journal of Mathematical Physics}, 43(9):4334--4340, 2002.

\bibitem[Sho04]{shor2004equivalence}
Peter~W Shor.
\newblock Equivalence of additivity questions in quantum information theory.
\newblock {\em Communications in Mathematical Physics}, 246(3):453--472, 2004.

\bibitem[Sti55]{stinespring1955positive}
W.~Forrest Stinespring.
\newblock {Positive functions on {$C^*$}-algebras}.
\newblock {\em Proc. Amer. Math. Soc.}, 6:211--216, 1955.

\bibitem[SW97]{schumacher1997sending}
Benjamin Schumacher and Michael~D Westmoreland.
\newblock Sending classical information via noisy quantum channels.
\newblock {\em Physical Review A}, 56(1):131, 1997.

\bibitem[Wei78]{weingarten1978asymptotic}
Don Weingarten.
\newblock Asymptotic behavior of group integrals in the limit of infinite rank.
\newblock {\em Journal of Mathematical Physics}, 19(5):999--1001, 1978.

\bibitem[WH02]{werner2002counterexample}
Reinhard~F Werner and Alexander~S Holevo.
\newblock Counterexample to an additivity conjecture for output purity of
  quantum channels.
\newblock {\em Journal of Mathematical Physics}, 43(9):4353--4357, 2002.

\bibitem[Wil17]{wilde2017quantum}
Mark~M Wilde.
\newblock {\em Quantum information theory}.
\newblock Cambridge University Press, 2017.

\end{thebibliography}
\bibliographystyle{alpha}

\end{document}